\documentclass[eqsecnum,superscriptaddress,nofootinbib,11pt]{revtex4}
\usepackage{amsmath}
   \usepackage{bm}
\usepackage{amsthm,amscd}
\usepackage{mathrsfs}
\usepackage{verbatim}
\usepackage{amsfonts,amsmath,latexsym,amssymb,txfonts,bm}
\usepackage[american]{babel}
\usepackage{dsfont}
\usepackage[dvips]{graphicx}
\usepackage{latexsym}
\usepackage{epsfig}
\newcommand{\diff}{\mathrm{d}}

\def\beq{\begin{eqnarray}}
\def\eeq{\end{eqnarray}}
\newcommand{\nn}{\nonumber}

\newtheorem{lemma}{Lemma}[]
\newtheorem{theorem}{Theorem}[]
\newtheorem{corollary}{Corollary}[]

\begin{document}
\title{Analytic Continuation of the Doubly-periodic Barnes Zeta Function}

\author{Guglielmo Fucci\footnote{Electronic address: fuccig@ecu.edu}}
\affiliation{Department of Mathematics, East Carolina University, Greenville, NC 27858 USA}

\author{Klaus Kirsten\footnote{Electronic address: Klaus\textunderscore Kirsten@Baylor.edu}}
\affiliation{Department of Mathematics, Baylor University, Waco, TX 76798 USA}

\date{April 16, 2013}
\vspace{2cm}
\begin{abstract}

The aim of this work is to study the analytic continuation of the doubly-periodic Barnes zeta function.
By using a suitable complex integral representation as a starting point we find the meromorphic extension of the doubly periodic Barnes zeta function
to the entire complex plane in terms of a real integral containing the Hurwitz zeta function and the first Jacobi theta function.
These allow us to explicitly give expressions for the derivative at all non-positive integer points.

\end{abstract}
\maketitle

\section{Introduction}

The Barnes zeta function, introduced for the first time in \cite{barnes1901,barnes1904}, represents a higher dimensional generalization of the Hurwitz zeta function
\beq \zeta_H (s,a) = \sum_{n=0}^\infty (n+a)^{-s} \quad \quad \mbox{for }\Re s>1. \label{zetahur}\eeq
Namely, let $s\in\mathbb{C}$, $\mu\in\mathbb{R}_{+}$, and ${\bf r}\in\mathbb{R}^{d}_{+}$. For $\Re s>d$ the Barnes zeta function is defined through the series
\begin{equation}
  \zeta_{\mathcal{B}}(s,\mu |{\bf r})=\sum_{{\bf n}\in\mathbb{N}_{0}^{d}}(\mu+{\bf n}\cdot{\bf r})^{-s}\;,
\end{equation}
and it can be analytically continued in a unique way to a meromorphic function in the entire complex plane possessing only simple poles at $s=1,2,...,d.$
In this work we shall be interested in the meromorphic extension of a function closely related to the two dimensional Barnes zeta function.
Let $b,c\in\mathbb{R}_{+}$, $\Omega_{m,n}=ibm+cn$ with $m,n\in\mathbb{Z}$, and $a\in\mathbb{C}\backslash\Omega_{m,n}$. For $\Re s>2$ we consider the following zeta function
\beq
\zeta (s,a,b,c)= \sum_{(m,\,n)\in\mathbb{Z}^{2}}(a+\mathrm{i}bm+cn)^{-s}\; ,\label{1}
\eeq
which is the kind of zeta function resulting from Dirac operators on the two-torus as considered in generalized Thirring models \cite{sachs96}.
The zeta function defined above is doubly-periodic with respect to the variable $a$. In fact for $k\in\mathbb{Z}$ one has the relations
$\zeta(s,a+kc,b,c)=\zeta(s,a,b,c)$ and $\zeta(s,a+\mathrm{i}kb,b,c)=\zeta(s,a,b,c)$, and due to this double-periodicity we assume, without loss
of generality, that $0<\Re a <c$ and $0<\Im a <b$. Since the zeta function (\ref{1}) is analytic in the semi-plane $\Re s>2$ and the ratio
of the two periods $c$ and $\mathrm{i}b$ is not a real number, one can conclude that (\ref{1}) belongs to the class of elliptic functions \cite{whitt27}.

The main idea of the present work is to represent the doubly-periodic Barnes zeta function (\ref{1}) for $\Re s>2$ in terms of a contour integral in the complex plane.
The desired analytic continuation of (\ref{1}) to the region $\Re s\leq 2$ is then achieved by a suitable deformation of the integration contour.
This process yields an expression for (\ref{1}) valid in the entire complex plane in terms of an integral over the interval $[0,1]$ of the Hurwitz zeta function and
the logarithmic derivative of the first Jacobi theta function. The analytically continued expression for $\zeta(s,a,b,c)$ will allow us to very easily
compute its values at all integer points, $s\in\mathbb{Z}$. In addition, we will also provide an explicit expression for the derivative of $\zeta(s,a,b,c)$ with respect to $s$
at all non-positive integer points.

We would like to point out that one of the main advantages of our study of the analytic continuation
of (\ref{1}) is that its double-periodicity property remains {\it manifest} in all the formulae. This is an aesthetically pleasing feature that can also be desirable if one wishes to
implement these expressions in a computer program for numerical evaluations.
The study of a zeta function closely related to the doubly-periodic Barnes zeta function considered here has appeared in \cite{elizalde07} where a method for obtaining its analytic continuation
was used which differs from the one we employ in this work.

The outline of the paper is as follows. In the next section we construct a contour integral representation for $\zeta(s,a,b,c)$ valid for $\Re s>2$.
From this representation we perform, in Section \ref{sec3}, the analytic continuation in the entire complex plane. Section \ref{sec4}
is devoted to the explicit computation of the derivative of $\zeta(s,a,b,c)$ with respect to the first variable at all negative integer points and at $s=0$.

\section{Contour Integral Representation of $\zeta (s,a,b,c)$}

As we have already mentioned in the Introduction, the main idea of our approach is to represent the doubly-periodic Barnes zeta function (\ref{1})
in terms of a contour integral. Since by assumption $c\neq 0$ we introduce, for convenience, the function
\beq\label{1b}
f_{m}(n,s) = \frac 1 {(\alpha_{m} +n)^s}\;,
\eeq
with $n\in\mathbb{Z}$ and $\alpha_{m} = \frac a c +\textrm{i}\frac b c m\in \mathbb{C}\backslash\{0\}$, where $m\in\mathbb{Z}$.
Obviously, in terms of the newly introduced function $f_{m}(n,s)$, the zeta function (\ref{1}) reads
\begin{equation}\label{1a}
  \zeta(s,a,b,c)=c^{-s}\sum_{m\in\mathbb{Z}}\sum_{n\in\mathbb{Z}}f_{m}(n,s)\;.
\end{equation}
By utilizing Cauchy's residue theorem we rewrite the sum over the index $n$ in (\ref{1a}) as a contour integral. More precisely one has
\beq\label{1c}
\sum_{n\in\mathbb{Z}}\frac 1 {(\alpha_{m} +n)^s}=
\sum_{n\in\mathbb{Z}} f_{m}(n,s) =\frac 1 {2\pi {\rm i}} \int_\Gamma \diff z \,\,
f_{m}(z,s) \,\, \pi \cot (\pi z)\;,
\eeq
where $\Gamma$ is a contour that encloses counterclockwise all the (simple) poles of the function $\pi \cot (\pi z)$.
Let us point out that the representation (\ref{1c}) is well defined in the region $\Re s>2$.

Before specifying the integration contour $\Gamma$ in detail, we would like to observe that the function $f_{m}(z,s)$, obtained from (\ref{1b}) by replacing $n$ with $z$,
possesses branch cuts extending from the points $z=-\alpha_{m}$. The exact position of the cut will depend explicitly on the
summation index $m$. First, note that from the assumptions stated below (\ref{1}) one obtains the inequalities $b/c>0$, $\Re (a/c)>0$, and $\Im (a/c)>0$.
This allows us to conclude that $\Re\alpha_{m}>0$ for all $m$, $\Im\alpha_{m}>0$ for $m\geq 0$ and $\Im\alpha_{m}<0$ for $m<0$.
The last remark shows that the cut lies in the lower half complex plane when
$m\geq 0$ and in the upper half complex plane when $m\leq -1$. The contour $\Gamma$
has to be chosen in such a way to enclose only the poles of $\cot (\pi z)$ but not the branch points of the function $f_{m}(z,s)$. More precisely, the contour is the union $\Gamma=\Gamma_+\cup\Gamma_-$ where $\Gamma_{+}$ satisfies the property $0<\Im \Gamma_{+}<-\Im \alpha_{m}$ for $m\leq -1$ while $\Gamma_{-}$ satisfies $\Im \alpha_{m}<\Im \Gamma_{-}<0$ for $m\geq 0$ (which simply means
that the contour is closer to the real axis than the cut). The contour $\Gamma$ is depicted in Figure \ref{fig1} with $-\alpha_{\geq 0}$ and $-\alpha_{\leq -1}$ denoting representatives of the set of branch points with index $m\geq 0$ and $m\leq -1$, respectively.

\begin{figure}[h!]
\centering
\includegraphics[scale=0.6,trim=0cm 6cm 0cm 8cm, clip=true, angle=0]{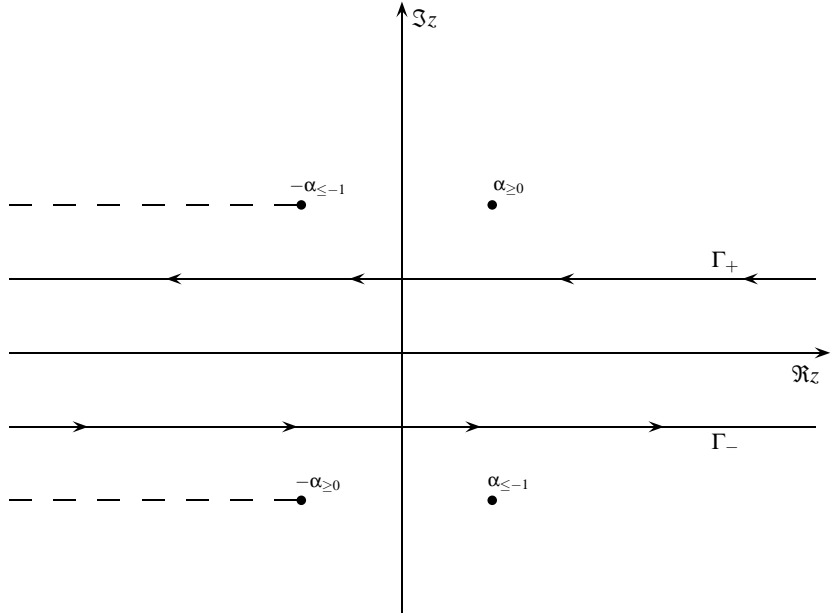}
\caption{Integration contour used for the representation (\ref{1c})}\label{fig1}
\end{figure}

With the contour of integration completely determined we can express the left hand side of (\ref{1c}) as a sum of two contributions
\begin{equation}\label{1d}
  \sum_{n\in\mathbb{Z}}\frac 1 {(\alpha_{m} +n)^s}=\frac 1 {2\pi {\rm i}} \int_{\Gamma_{+}} \diff z \,\,
f_{m}(z,s) \,\, \pi \cot (\pi z)+\frac 1 {2\pi {\rm i}} \int_{\Gamma_{-}} \diff z \,\,
f_{m}(z,s) \,\, \pi \cot (\pi z)\;.
\end{equation}

The next step of our approach is to rewrite the integrand in (\ref{1d}) in a way that will allow a separate treatment of the integral over
$\Gamma_{+}$ and over $\Gamma_{-}$. By utilizing the complex exponential representation of the function $\cot (\pi z)$ and after some straightforward
algebraic manipulations we obtain
\begin{equation}\label{rep1}
  \cot (\pi z)=-{\rm i} \left(1+\frac 2 {e^{-2\pi {\rm i}z}-1} \right)\;,\quad\textrm{for}\quad \Im z>0\;,
\end{equation}
and
\begin{equation}\label{rep2}
  \cot (\pi z)={\rm i}\left(1+\frac 2 {e^{2\pi {\rm i}z}-1}\right)\;,\quad\textrm{for}\quad \Im z<0\;.
\end{equation}
Since $\Im \Gamma_{+}>0$ and $\Im \Gamma_{-}<0$, we use the representation (\ref{rep1}) for the integral over $\Gamma_{+}$ and the
representation (\ref{rep2}) for the integral over $\Gamma_{-}$. Proceeding in this fashion leads to the results
\beq\label{2}
\frac 1 {2\pi {\rm i}} \int_{\Gamma_+} \diff z \,\, f_{m}(z,s)  \,\,
    \pi \cot (\pi z) = -\frac 1 2 \int_{\Gamma_+} \diff z \,\, f_{m}(z,s)
 -\int_{\Gamma_+} \diff z \,\, \frac{f_{m}(z,s) }{e^{-2\pi {\rm i}z}-1} , 
\eeq
and
\beq\label{2a}
\frac 1 {2\pi {\rm i}} \int_{\Gamma_-} \diff z \,\, f_{m}(z,s)  \,\,
    \pi \cot (\pi z) = \frac 1 2 \int_{\Gamma_-} \diff z \,\, f_{m}(z,s)
 +\int_{\Gamma_-} \diff z \,\, \frac{f_{m}(z,s) }{e^{2\pi {\rm i}z}-1} . 
\eeq
The next step is the deformation of the integration contours which is the subject of the following Lemma.
\begin{lemma}\label{lemma1}
  Let $\Gamma=\Gamma_{+}\cup\Gamma_{-}$ be the integration contour described above, and $z=u+\alpha_{m}$ with $m\in\mathbb{Z}$.
  Then, in the semi-plane $\Re s>2$, one obtains
  \begin{equation}\label{claim}
    \frac 1 {2\pi {\rm i}} \int_\Gamma \diff z \,\,
f_{m}(z,s) \,\, \pi \cot (\pi z)=-\int_\gamma \diff u \,\, \frac{u^{-s}}{e^{-2\pi {\rm i}(u-\alpha_{\leq-1})}-1}+\int_\gamma \diff u \,\, \frac{u^{-s}}{e^{2\pi {\rm i}(u-\alpha_{\geq 0})}-1}\;,
  \end{equation}
  where $\gamma$ is a contour enclosing in the clockwise direction the negative real axis including the point $u=0$.
\end{lemma}
\begin{proof} The proof of this result is based on a suitable deformation of the integration contour $\Gamma$.
Before deforming the contour, we focus on the first term on the left hand side of (\ref{2}) and (\ref{2a}). From the definition of the function $f_{m}(z,s)$ in
(\ref{1b}) we can write
\begin{equation}
  \int_{\Gamma_+} \diff z \,\, f_{m}(z,s) \,\, = \int_{\infty}^{-\infty}\diff x \,\,
   (x+\alpha_{m}+{\rm i}\Im\Gamma_{+})^{-s}\;.
\end{equation}
Since, by assumption, the analysis is restricted to the region $\Re s>2$ and since $\Re\alpha_{m}>0$ one can conclude that for all $m$
\begin{equation}
  \int_{\Gamma_+} \diff z \,\, f_{m}(z,s)=0\;.
\end{equation}
By using a similar argument one can prove that for all $m$ the following relation holds
\beq
\int_{\Gamma_-} \diff z \,\, f_{m}(z,s)=\int_{-\infty}^{+\infty}\diff x \,\,
   (x+\alpha_{m}-{\rm i}\Im\Gamma_{-})^{-s}=0\;.
\eeq

Let us consider next the second integral on the right hand side of (\ref{2}), which can be rewritten as
\begin{equation}
 \int_{\Gamma_+} \diff z \,\, \frac{f_{m}(z,s) }{e^{-2\pi {\rm i}z}-1}=\int_{\infty}^{-\infty}\diff x\,\,(x+\alpha_{m}+{\rm i}\Im\Gamma_{+})^{-s}\frac{1}{e^{-2\pi{\rm i}(x+\alpha_{m})}e^{2\pi\Im\Gamma_{+}}-1}\;,
\end{equation}
and is convergent for $\Re s>2$. If the branch cut extends from the points $-\alpha_{\geq 0}$, which simply means that it lies in the lower half plane,
the contour $\Gamma_+$ can be shifted away to infinity in the upper half plane and does not
contribute. In fact
\begin{eqnarray}
 \left|(x+\alpha_{\geq 0}+{\rm i}\Im\Gamma_{+})^{-s}\frac{1}{e^{-2\pi{\rm i}(x+\alpha_{\geq 0})}e^{2\pi\Im\Gamma_{+}}-1}\right|&\leq&
\frac{1}{e^{2\pi\Im\Gamma_{+}}-1}\left[(x+\alpha_{\geq 0})^{2}+(\Im\Gamma_{+})^{2}\right]^{-\frac{s}{2}}\longrightarrow 0\;,
\end{eqnarray}
as $\Im\Gamma_{+}\to\infty$. If, on the other hand, the branch cut extends from the points $-\alpha_{\leq-1}$ in the upper half plane, we shift the contour $\Gamma_{+}$ around the branch cut
as shown in Figure \ref{fig2}.

The second integral over $\Gamma_{-}$ in (\ref{2a}) is convergent for $\Re s>2$ and can be expressed as
\begin{equation}\label{5}
\int_{\Gamma_-} \diff z \,\, \frac{f_{m}(z,s) }{e^{2\pi {\rm i}z}-1}=\int_{-\infty}^{\infty}\diff x\,\,(x+\alpha_{m}-{\rm i}\Im\Gamma_{-})^{-s}\frac{1}{e^{2\pi{\rm i}(x+\alpha_{m})}e^{-2\pi\Im\Gamma_{-}}-1}\;.
\end{equation}
By using arguments similar to the ones outlined above one can prove that if the cut extends from the branch points $-\alpha_{\leq-1}$ in the
upper half plane, then we can shift the contour $\Gamma_{-}$ away to infinity, namely $\Im\Gamma_{-}\to-\infty$, and the integral (\ref{5}) vanishes.
If, instead, the branch cut extends from the points $-\alpha_{\geq0}$, in the lower half plane, we shift the contour $\Gamma_{-}$
as shown in Figure \ref{fig2}.

\begin{figure}[h!]
\centering
\includegraphics[scale=0.6,trim=0cm 6cm 0cm 8cm, clip=true, angle=0]{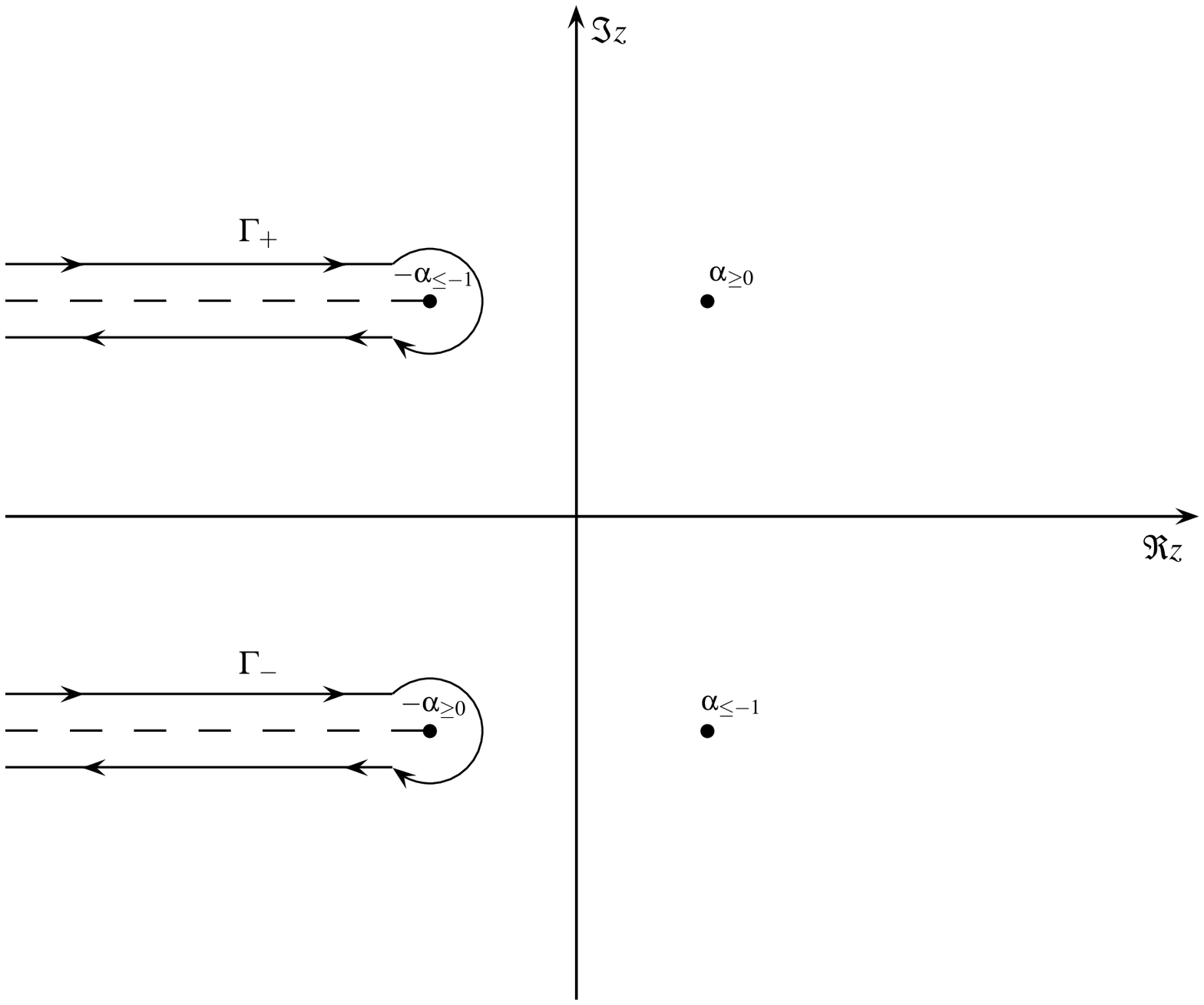}
\caption{}\label{fig2}
\end{figure}

We can therefore conclude that for $m\leq-1$ only the integral over the deformed $\Gamma_{+}$ gives a non-vanishing contribution
and by making the substitution $u=z+\alpha_{\leq-1}$ we obtain
\begin{equation}\label{6}
  \frac 1 {2\pi {\rm i}} \int_{\Gamma_+} \diff z \,\, f_{m}(z,s)  \,\,
    \pi \cot (\pi z) =-\int_\gamma \diff u \,\, \frac{u^{-s}}{e^{-2\pi {\rm i}(u-\alpha_{\leq-1})}-1}\;,
\end{equation}
where $\gamma$ is a contour enclosing in the clockwise direction the negative real axis including the point $u=0$ as shown in Figure \ref{fig3}.
By using a similar argument, when $m\geq0$, we obtain
\begin{equation}\label{7}
  \frac 1 {2\pi {\rm i}} \int_{\Gamma_-} \diff z \,\, f_{m}(z,s)  \,\,
    \pi \cot (\pi z) =\int_\gamma \diff u \,\, \frac{u^{-s}}{e^{2\pi {\rm i}(u-\alpha_{\geq 0})}-1}\;.
\end{equation}
By substituting the results (\ref{6}) and (\ref{7}) in the relation (\ref{1d}) we find the claim (\ref{claim}).
\end{proof}

\begin{figure}[h!]
\centering
\includegraphics[scale=0.6,trim=0cm 6cm 0cm 12cm, clip=true, angle=0]{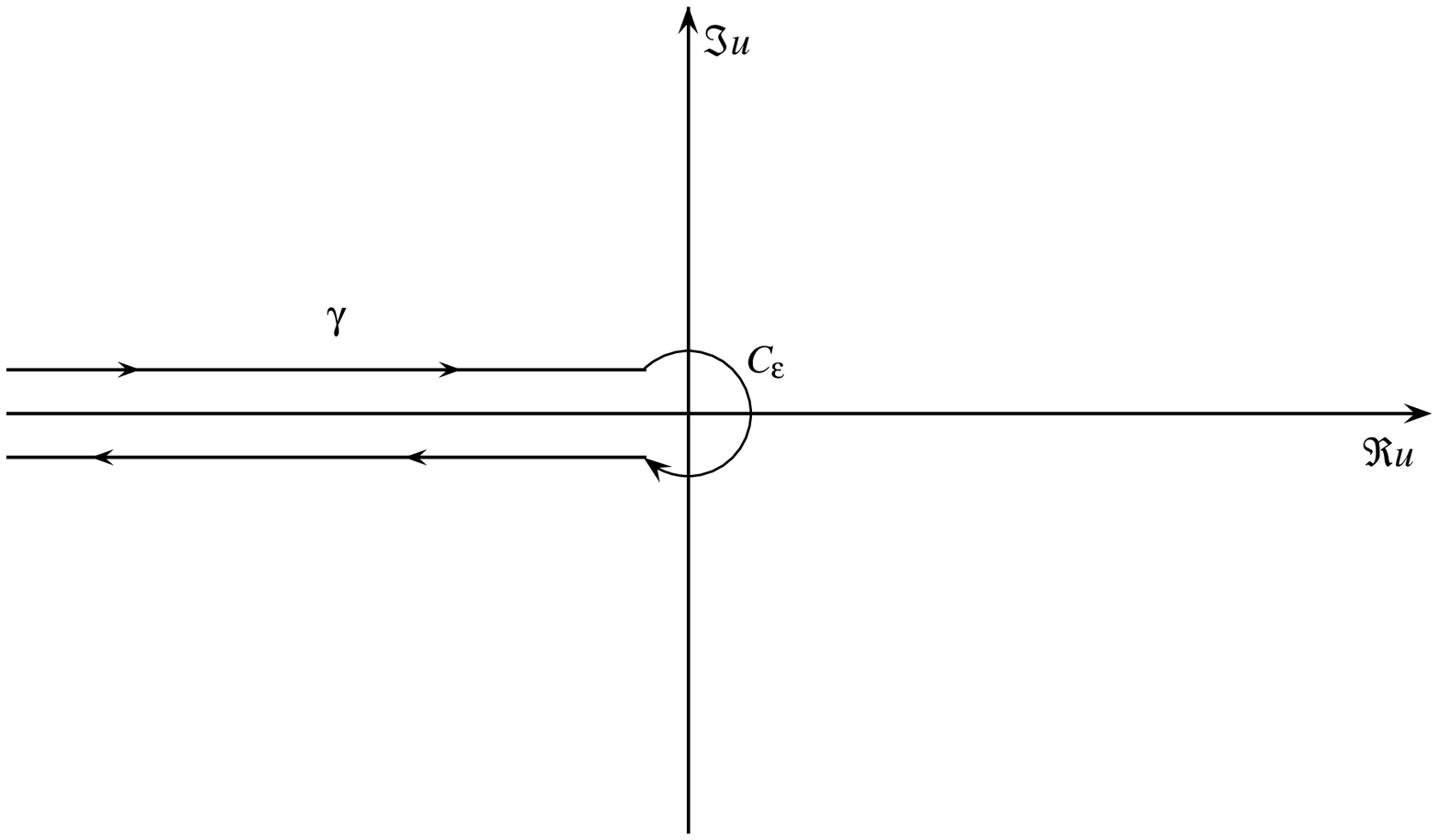}
\caption{}\label{fig3}
\end{figure}
Lemma 1 immediately allows to rewrite the doubly-periodic Barnes zeta function. First note, that
from the expressions (\ref{1b}) and (\ref{1a}) we can write
\begin{equation}\label{8}
  \zeta(s,a,b,c)=\frac{c^{-s}}{2\pi}\sum_{m\in\mathbb{Z}}\int_{\Gamma}\diff z\,f_{m}(s,z)\cot(\pi z)\; ,
\end{equation}
which by using
Lemma \ref{lemma1} gives
\beq\label{9}
\zeta  (s,a,b,c)&=&-c^{-s} \sum_{m=-1}^{-\infty} \int_\gamma \diff u \,\,
     \frac{u^{-s}}{e^{-2\pi {\rm i} \left(u-\frac{a}{c}-{\rm i}\frac{b}{c}m\right)}-1}+ c^{-s} \sum_{m=0}^\infty \int_\gamma \diff u \,\,
       \frac{u^{-s}}{e^{2\pi {\rm i} \left(u-\frac{a}{c}-{\rm i}\frac{b}{c}m\right)}-1}  \nn\\
   &=& c^{-s} \int_\gamma \diff u \,\,
        \frac{u^{-s}}{e^{2\pi {\rm i} \left(u-\frac{a}{c}\right)}-1}\nonumber\\
        &+& c^{-s} \sum_{m=1}^{\infty}\int_\gamma \diff u \,\,
        u^{-s}\left\{\frac 1 {e^{2\pi {\rm i}\left(u-\frac{a}{c}\right)+2\pi \frac{b}{c}m}-1}-\frac 1 {e^{-2\pi {\rm i}\left(u-\frac{a}{c}\right)+2\pi \frac{b}{c}m}-1} \right\}\;,
\eeq
where the integral in the second line represents the contribution due to $m=0$.
By introducing a new variable $q\in\mathbb{R}$ defined as $q=e^{-\pi b/c}$ and by combining the two exponentials in the last line of (\ref{9}) we find the expression
\beq\label{10}
\zeta (s,a,b,c) &=& c^{-s} \int_\gamma \diff u \,\, \frac{u^{-s}}
              {e^{2\pi i(u-a/c)}-1} \nn\\
     & &+c^{-s}\sum_{m=1}^\infty \int_\gamma \diff u \,\,u^{-s}
     \left[-2{\rm i} \sin \left( 2\pi \left[
      u-\frac a c\right]\right) \frac{q^{2m}}
         {1-2\cos \left( 2\pi \left[u-\frac a c\right]\right) q^{2m}
           +q^{4m}}\right]\;,
\eeq
which, once again, is valid for $\Re s>2$. The second integral appearing in (\ref{10}), although rather cumbersome,
can be expressed in terms of a simple special function.
In fact, let $\tau\in\mathbb{C}$ with $|\tau|<1$ and $z\in\mathbb{C}$; the first
Jacobi theta function $\vartheta_1 (z;\tau)$ has the following representation as an infinite product \cite{gradshtein07}
\beq\label{11}
\vartheta_1 (z;\tau) = 2 \tau^{1/4} \sin z \prod_{n=1}^\infty \left(1-2\tau^{2n}\cos (2z)
          +\tau^{4n}\right) \left(1-\tau^{2n}\right)\;.
\eeq
By taking the derivative of the logarithm of (\ref{11}) and by using $q$ introduced below (\ref{9}) for $\tau$, we obtain the formula \cite{erdelyi53}
\beq\label{12}
\frac \diff {\diff z} \ln \vartheta_1 (z;q) = \frac{\vartheta^{\prime}_{1}(z;q)}{\vartheta_{1}(z;q)} = \cot z +4 \sin (2z) \sum_{n=1}^\infty
      \frac{q^{2n}}{1-2q^{2n}\cos(2z) +q^{4n}}\;.
\eeq
Since $|q|<1$, for any finite $z$ the sum in (\ref{12}) is absolutely convergent \cite{shen}. This implies, in particular, that
the convergence of the series in (\ref{12}) is uniform in $z$ in the region $|\Im z|<\pi b/c$. The last remark justifies the interchange
of the sum and the integral in (\ref{10}) to obtain, by using the expression (\ref{12}) with the substitution $z\to\pi(u-a/c)$, the representation
\beq\label{14}
\zeta (s,a,b,c) &=& c^{-s} \int_\gamma \diff u \frac{u^{-s}}
       {e^{2\pi i (u-a/c)}-1}+\frac{{\rm i}}{2}c^{-s}\int_\gamma \diff u \,\, u^{-s}\cot \left( \pi \left[u-\frac a c \right]\right)\nn\\
       & &-\frac {\rm i} {2\pi} c^{-s} \int_\gamma \diff u \,\, u^{-s}
         \frac \diff {\diff u} \ln \vartheta_1 \left( \pi \left[u-\frac a c\right]
        ;q\right)
        \;,
\eeq
where the contour $\gamma$ is chosen so that $|\Im\gamma|<\pi b/c$ in order to allow the use of (\ref{12}). By combining the first two integrals in (\ref{14}) and by using the relation
\begin{equation}
  \frac{1}{e^{2\pi{\rm i}x}-1}+\frac{{\rm i}}{2}\cot\pi x=-\frac{1}{2}\;,
\end{equation}
we find
\beq\label{15}
\zeta (s,a,b,c) = -\frac {\rm i} {2\pi} c^{-s} \int_\gamma \diff u \,\, u^{-s}
    \frac \diff {\diff u} \ln \vartheta _1 \left( \pi \left[ u-\frac a c \right]
     ;q\right)-\frac{c^{-s}}{2}\int_\gamma \diff u \,\, u^{-s}\;.
\eeq
Since $\Re s>2$, the second integral in (\ref{15}) gives a vanishing contribution and we are left with the following compact representation
\beq \label{16}
\zeta (s,a,b,c) =-\frac {\rm i} {2\pi} c^{-s} \int_\gamma \diff u \,\, u^{-s}
\frac \diff {\diff u} \ln \vartheta_1 \left( \pi \left[u-\frac a c\right]
;q\right)\;.
\eeq
We would like to point out that the last integral representation preserves the double-periodicity property of the original sum in (\ref{1}).
In fact, let $m, n\in\mathbb{Z}$. By using (\ref{16}) we write
\begin{equation}\label{16a}
\zeta(s,a+nc+{\rm i}mb,b,c)=-\frac {\rm i} {2\pi} c^{-s} \int_\gamma \diff u \,\, u^{-s}
\frac \diff {\diff u} \ln \vartheta_1 \left( \pi \left[u-\frac a c\right]-n\pi-{\rm i}m\pi\frac{b}{c}
;q\right)\;.
\end{equation}
Now, let $z_{m,\,n}=(m+n\tau)\pi$  denote the vertices of the fundamental parallelogram in the $z$-plane. Then the first Jacobi theta function
is quasi-periodic on the lattice \cite{erdelyi53}
\begin{equation}\label{prop1}
  \vartheta_{1}(z+(m+n\tau)\pi;\tau)=(-1)^{m+n}e^{-{\rm i}(2nz+\pi n^{2}\tau)}\vartheta_{1}(z;\tau)\;.
\end{equation}
By utilizing (\ref{prop1}) in (\ref{16a}) we obtain
\begin{equation}
  \zeta(s,a+nc+{\rm i}mb,b,c)=-nc^{-s}\int_\gamma\diff u \,\, u^{-s}+\zeta(s,a,b,c)=\zeta(s,a,b,c)\;,
\end{equation}
where, as before, the last equality follows from the fact that since $\Re s>2$ the integral over $\gamma$ vanishes.

The representation (\ref{16}) allows to easily compute $\zeta (s,a,b,c)$ for $s=n$, $n\in\mathbb{N}$, $n\geq 2$. For these values of $s$, the contour encloses a singularity at $u=0$ of order $n$ and the
residue theorem immediately shows
\begin{eqnarray}\label{18}
 \zeta (n,a,b,c) &=& -\frac {\rm i} {2\pi} c^{-n} \int_{C_{\varepsilon}} \diff u \,\, u^{-n}
\frac \diff {\diff u} \ln \vartheta_1 \left( \pi \left[u-\frac a c\right];q\right)\nn\\
&=&-\left(\frac \pi c \right) ^n \frac 1 {(n-1)!}
      \frac{\diff^n}{\diff u^n} \ln \vartheta_1 (u;q) \Big|_{u=-\pi \frac{a}{c}} \nn\\
&=&\left(\frac \pi c \right) ^n \frac 1 {(n-1)!}
      \frac{\diff^n}{\diff u^n} \ln \vartheta_1 (u;q) \Big|_{u=\pi \frac{a}{c}} \; .
\end{eqnarray}
We next use (\ref{16}) to construct the representation of $\zeta (s,a,b,c)$ valid in the whole complex plane.

\section{Analytic Continuation}\label{sec3}

The integration contour $\gamma$ in (\ref{16}) consists of a union of three distinct paths, namely $\gamma=\gamma_{+}\cup C_{\varepsilon}\cup\gamma_{-}$, where $C_{\varepsilon}$
is the circular portion of radius $\varepsilon$ centered at the origin, $\gamma_{+}$ is the straight path positioned at a distance $\delta$ above the negative real axis, and
$\gamma_{-}$ represents the straight path positioned at a distance $\delta$ below the negative real axis. Furthermore, for later use, we denote by $-\tilde{\varepsilon}$
the projection on the negative real axis of the intersection of $\gamma_{+}$ (or $\gamma_{-}$) with the circular portion $C_{\varepsilon}$ of the integration path.
This remark allows us to rewrite (\ref{16}) as
\begin{equation}\label{17}
  \zeta (s,a,b,c)=-\frac {\rm i} {2\pi} c^{-s} \int_{C_{\varepsilon}} \diff u \,\, u^{-s}
\frac \diff {\diff u} \ln \vartheta_1 \left( \pi \left[u-\frac a c\right]
;q\right)-\frac {\rm i} {2\pi} c^{-s} \int_{\gamma_{+}\cup\gamma_{-}} \diff u \,\, u^{-s}
\frac \diff {\diff u} \ln \vartheta_1 \left( \pi \left[u-\frac a c\right]
;q\right)\;.
\end{equation}
The last representation is a suitable starting point from which we can proceed with the analytic continuation to the region $\Re s\leq 2$.

The first term in (\ref{17}), namely the integral along $C_\varepsilon$, is left unchanged; it will be dealt with later.
The second term in (\ref{17}) can be expressed as a sum $\mathcal{I}_{\gamma_{+}}+\mathcal{I}_{\gamma_{-}}$,
where $\mathcal{I}_{\gamma_{+}}$ represents the integral over the path $\gamma_{+}$ and $\mathcal{I}_{\gamma_{-}}$ the one over the path $\gamma_{-}$.
We will present details for the integral $\mathcal{I}_{\gamma_{+}}$, $\mathcal{I}_{\gamma_-}$ follows accordingly.
We parameterize the path $\gamma_{+}$ as $u={\rm i}\delta +x$, with $x\in(-\infty,-\tilde\varepsilon ]$, and rewrite
the integral over $\gamma_{+}$ as follows
\begin{eqnarray}\label{19}
  \mathcal{I}_{\gamma_{+}}&=&-\frac {\rm i} {2\pi} c^{-s} \int_{\gamma_{+}} \diff u \,\, u^{-s}
\frac \diff {\diff u} \ln \vartheta_1 \left( \pi \left[u-\frac a c\right]
;q\right)\nonumber\\
&=&-\frac {\rm i} {2\pi} c^{-s}\int_{-\infty}^{-\tilde\varepsilon}\diff x\,\,(x+{\rm i}\delta )^{-s}\frac \diff {\diff x} \ln \vartheta_1 \left( \pi \left[x+{\rm i}\delta -\frac a c\right]
;q\right)\nonumber\\
&=&\frac {\rm i} {2\pi} c^{-s}\int^{\infty}_{\tilde\varepsilon}\diff x\,\,({\rm i}\delta -x)^{-s}\frac \diff {\diff x} \ln \vartheta_1 \left( \pi \left[{\rm i}\delta -x-\frac a c\right]
;q\right)\;.
\end{eqnarray}
The integral in the last line of (\ref{19}) is well defined at the lower limit of integration. To establish the convergence at the upper limit of integration it seems
we would need to analyze the asymptotic behavior of the integrand as $x\to\infty$. However, this analysis
can be avoided by using the following quasi-periodicity property
\beq\label{20}
\vartheta_1 (u+\pi ;q) = -\vartheta_1 (u;q)\;,
\eeq
which is obtained by setting $n=0$ and $m=1$ in the more general formula (\ref{prop1}).

In order to exploit the above property we represent $\mathcal{I}_{\gamma_{+}}$ as
\begin{equation}\label{21}
  \mathcal{I}_{\gamma_{+}}=\frac {\rm i} {2\pi} c^{-s}\sum_{n=0}^{\infty}\int_{n+\tilde\varepsilon}^{n+1+\tilde\varepsilon}
  \diff x\,\,({\rm i}\delta -x)^{-s}\frac \diff {\diff x} \ln \vartheta_1 \left( \pi \left[{\rm i}\delta -x-\frac a c\right]
;q\right)\;,
\end{equation}
and perform the change of variables $x=y+n$ to obtain
\begin{equation}\label{22}
  \mathcal{I}_{\gamma_{+}}=\frac {\rm i} {2\pi} c^{-s}\sum_{n=0}^{\infty}\int_{\tilde\varepsilon}^{1+\tilde\varepsilon}
  \diff y\,\,({\rm i}\delta -y-n)^{-s}\frac \diff {\diff y} \ln \vartheta_1 \left( \pi \left[{\rm i}\delta -y-n-\frac a c\right]
;q\right)\;.
\end{equation}
Next we send $\delta \to 0$, in which case $\tilde \varepsilon \to \varepsilon$. By using the quasi-periodicity (\ref{20}) in the expression (\ref{22}), we then have
\begin{equation}\label{23}
  \mathcal{I}_{\gamma_{+}}=\frac {\rm i} {2\pi} c^{-s}e^{-{\rm i}\pi s}\sum_{n=0}^{\infty}\int_{\varepsilon}^{1+\varepsilon}
  \diff y\,\,(y+n)^{-s}\frac \diff {\diff y} \ln \vartheta_1 \left( \pi \left[-y-\frac a c\right]
;q\right)\;.
\end{equation}
We note, at this point, that the series $\sum_{n=0}^{\infty}(y+n)^{-s}$ converges uniformly in $y$ in the interval $[\varepsilon ,1+\varepsilon ]$ for $\Re s>1$. This allows us to rewrite (\ref{23})
as
\begin{equation}\label{24}
  \mathcal{I}_{\gamma_{+}}=\frac {\rm i} {2\pi} c^{-s}e^{-{\rm i}\pi s}\int_{\varepsilon}^{1+\varepsilon }
  \diff y\,\,\zeta_{H}(s,y)\frac \diff {\diff y} \ln \vartheta_1 \left( \pi \left[-y-\frac a c\right]
;q\right)\;,
\end{equation}
where the Hurwitz zeta function $\zeta_{H}(s,x)$ in equation (\ref{zetahur}) has been used. As is well known, the Hurwitz zeta function
can be analytically continued in a unique way to a meromorphic function in the entire complex plane possessing only a simple pole with residue $1$ at the point $s=1$.

For the integral $\mathcal{I}_{\gamma_{-}}$ over the lower part of the contour we proceed similarly to obtain
\beq\label{25}
\mathcal{I}_{\gamma_{-}}=-\frac {\rm i} {2\pi} c^{-s} e^{{\rm i}\pi s} \int_\varepsilon^{1+\varepsilon} dy \,\,
\zeta_H (s;y) \frac {\diff} {\diff y} \ln \vartheta_1
        \left( \pi \left[-y -\frac a c \right];q\right)\;.
\eeq
By adding the contribution from $\mathcal{I}_{\gamma_{+}}$ and $\mathcal{I}_{\gamma_{-}}$ we obtain the result
\beq
\mathcal{I}_{\gamma_{+}\cup\gamma_{-}}=\mathcal{I}_{\gamma_{+}}+\mathcal{I}_{\gamma_{-}} = \frac{\sin (\pi s)} \pi c^{-s}
  \int_\varepsilon^{1+\varepsilon} dy \,\, \zeta_H (s;y) \frac d {dy} \ln\vartheta_1 \left(\pi
        \left[y+\frac a c\right];q\right) .\label{26}
\eeq
Noticing that the integral along $C_\varepsilon$ in (\ref{17}) is defined for $s\in \mathbb{C}$,
the above expression performs the analytic continuation of $\zeta(s,a,b,c)$ in terms of the Hurwitz zeta function $\zeta_{H}(s,y)$ to the full complex plane and it is valid for $\varepsilon > 0$ small enough. By choosing $\Re s <1$, the limit $\varepsilon \to 0$ can be performed and the integral along $C_\varepsilon$ in (\ref{17}) can be seen to vanish.

%
We can now summarize the results obtained in this Section as follows:
\begin{theorem}
  Let $\zeta(s,a,b,c)$ denote the doubly-periodic Barnes zeta function defined in (\ref{1}). \newline
  If $\Re s <1$, then
  \begin{equation}\label{claim1}
    \zeta(s,a,b,c)=\frac{\sin (\pi s)} \pi c^{-s}
  \int_0^1 dy \,\, \zeta_H (s;y) \frac \diff {\diff y} \ln\vartheta_1 \left(\pi
        \left[y+\frac a c\right];q\right)\;.
  \end{equation}
  If $s=n$ with $n\in\mathbb{N}_{+}-\{1\}$, then
  \begin{equation}\label{claim2}
    \zeta(n,a,b,c)=\left(\frac \pi c \right) ^n \frac 1 {(n-1)!}
      \frac{\diff^n}{\diff u^n} \ln \vartheta_1 (u;q) \Big|_{u=\pi \frac{a}{c}} \;.
  \end{equation}
\end{theorem}

Due to the prefactor $\sin (\pi s)/\pi$ in
the integral representation (\ref{claim1}), the substitution $s=-n$ with $n\in\mathbb{N}_{0}$ leads to a vanishing contribution of $\zeta(s,a,b,c)$ when the first argument is a negative integer including zero. This remark proves the following
\begin{corollary}
  Let $n\in\mathbb{N}_{0}$, then $\zeta(-n,a,b,c)=0$.
\end{corollary}
In order to find the representation of $\zeta (s,a,b,c)$ valid for $\Re s >1$, note that for $y\to 0$ we have the behavior
\begin{equation}
\zeta_H (s,y) = \frac 1 {y^s} + \zeta_R (s) + {\cal O} (y)\;,
\end{equation}
where the first term results from the $n=0$ contribution in (\ref{zetahur}). This is the term responsible for the restriction $\Re s <1$ found in the previous results. The analytic continuation of (\ref{claim1})
to the full plane is found by observing that
\begin{equation}\label{add1}
\zeta_H (s,y) = \frac 1 {y^s} + \zeta_H (s,y+1)\;.
\end{equation}
By substituting (\ref{add1}) in the representation (\ref{claim1}) we find that the second term in (\ref{add1}) gives a function holomorphic for $s\in\mathbb{C}$. The resulting contribution to (\ref{claim1}) coming from the first term in (\ref{add1})
gives, instead, the integral
\begin{equation}\label{add2}
I_o = \int\limits_0^1 dy \,\, y^{-s} \frac d {dy} \ln \vartheta_1 \left(\pi
        \left[y+\frac a c\right];q\right)\;,
\end{equation}
which is valid for $\Re s <1$. At this point, partial integration can be applied repeatedly to obtain a representation that extends expression 
(\ref{add2}) further to the right of $\Re s=1$. For example, after one partial integration (\ref{add2}) reads
\beq
I_0 &=& \frac 1 {1-s} \left\{ \left.\frac d {dy} \ln \vartheta_1 \left(\pi
        \left[y+\frac a c\right];q\right)\right|_{y=1} -\int\limits_0^1 dy \,\, y^{-s+1} \frac{d^2}{dy^2} \ln \vartheta_1 \left(\pi
        \left[y+\frac a c\right];q\right) \right\},\label{extrep}\eeq
which is now valid for $\Re s<2$. After $n$ partial integrations a representation valid for $\Re s < n+1$ is found and (\ref{claim2}) can be verified from there.
The above representation also shows
\begin{equation}\label{claim3}
   \zeta(1,a,b,c)=\frac \pi c
      \frac{\diff}{\diff u} \ln \vartheta_1 (u;q) \Big|_{u=\pi \frac{a}{c}}-\frac 1 c
\ln \frac{\vartheta_1 \left( \pi \left(1+\frac{a}{c}\right); q\right)}{\vartheta_1 \left(\pi \frac{a}{c};
q\right)} ,
\end{equation}
where the first term comes from $I_0$, and the second term from the pole of $\zeta (s;y+1)$ at $s=1$.
The second term can be simplified noting
\beq\label{32}
\ln \frac{\vartheta_1 \left( \pi \left[ 1 +\frac a c \right] ;q\right) }
        {\vartheta_1 \left( \pi\frac a c ;q\right) } = -i\pi\;,
\eeq
which can be proved starting from the property (\ref{20}). 

\section{The derivative of the doubly-periodic Barnes zeta function at negative integers}\label{sec4}

In addition to the results already presented, from the integral representation (\ref{claim1}) one can compute the derivative of the doubly-periodic Barnes zeta function with respect to the first variable.
Differentiating (\ref{claim1}) leads to the result
\begin{eqnarray}\label{27}
  \zeta'(s,a,b,c)&=&\cos(\pi s)c^{-s}
  \int_0^1 dy \,\, \zeta_H (s;y) \frac \diff {\diff y} \ln\vartheta_1 \left(\pi
        \left[y+\frac a c\right];q\right)\nonumber\\
        &-&\frac{\sin (\pi s)} \pi c^{-s}\ln c
  \int_0^1 dy \,\, \zeta_H (s;y) \frac \diff {\diff y} \ln\vartheta_1 \left(\pi
        \left[y+\frac a c\right];q\right)\nonumber\\
        &+&\frac{\sin (\pi s)} \pi c^{-s}
  \int_0^1 dy \,\, \zeta'_H (s;y) \frac \diff {\diff y} \ln\vartheta_1 \left(\pi
        \left[y+\frac a c\right];q\right)\;,
\end{eqnarray}
where, here and in the rest of this work, the prime indicates differentiation with respect to the variable $s$.
By setting $s=0$ in (\ref{27}) and by noting that for $n\in\mathbb{N}_{0}$ the following relation holds \cite{gradshtein07}
\begin{equation}\label{28}
  \zeta_{H}(-n,x)=-\frac{B_{n+1}(x)}{n+1}\;,
\end{equation}
with $B_{n}(x)$ denoting the Bernoulli polynomials defined in terms of the Bernoulli numbers $B_{k}$ as
\begin{equation}\label{29}
B_{n}(x)=\sum_{k=0}^{n}\binom{n}{k}B_{k}x^{n-k}\;,
\end{equation}
we find
\beq
\zeta ' (-n,a,b,c)= (-1)^{n+1}c^{n}\sum_{k=0}^{n+1}\binom{n+1}{k}\frac{B_{k}}{n+1}\int_0^1 \diff y \,\, y^{n-k+1} \frac \diff {\diff y}
         \ln \vartheta_1 \left(\pi \left[ y +\frac a c \right] ;q\right)\;.\label{pre}
\eeq

Integrating by parts (\ref{pre}) yields the result
\begin{eqnarray}\label{30}
\zeta ' (-n,a,b,c)&=& \frac{(-1)^{n+1}c^{n}}{n+1}\ln\vartheta_{1}\left(\pi\left[\frac{a}{c}+1\right];q\right)\sum_{k=0}^{n}\binom{n+1}{k}B_{k}
+\frac{(-1)^{n+1}c^{n}}{n+1}B_{n+1}\ln \frac{\vartheta_1 \left( \pi \left[ 1+\frac{a}{ c} \right];q
\right)} {\vartheta_1 \left( \pi \frac{a}{ c} ; q\right)}\nonumber\\
&-&\frac{(-1)^{n+1}c^{n}}{n+1}\sum_{k=0}^{n}\binom{n+1}{k}(n-k+1)B_{k}\int_0^1 \diff y \,\, y^{n-k}
         \ln \vartheta_1 \left(\pi \left[ y +\frac a c \right] ;q\right)\;.
\end{eqnarray}
At this point it is convenient to distinguish between two cases: $n=0$ and $n\geq 1$. The reason for this distinction lies in the relation \cite{gradshtein07}
\begin{equation}\label{31}
  \sum_{k=0}^{n}\binom{n+1}{k}B_{k}=\left\{\begin{array}{ll}
  0 & \textrm{for}\;\; n\geq 1\\
  B_{0} & \textrm{for}\;\; n=0\;.
  \end{array}\right.
\end{equation}
By exploiting the above relation and (\ref{32}), one obtains
\begin{eqnarray}\label{33}
  \zeta ' (0,a,b,c)= -\ln \vartheta_1 \left(\pi \frac a   c ; q\right)+\frac{{\rm i}\pi}{2}+\int_{0}^{1}\diff y\,\,
         \ln \vartheta_1 \left(\pi \left[ y +\frac a c \right] ;q\right)\;,
\end{eqnarray}
and, when $n\geq 1$,
\begin{eqnarray}\label{34}
  \zeta ' (-n,a,b,c)= {\rm i}\pi\frac{(-1)^{n}c^{n}}{n+1} B_{n+1}+\frac{(-1)^{n}c^{n}}{n+1}\sum_{k=0}^{n}\binom{n+1}{k}(n-k+1)B_{k}\int_{0}^{1}\diff y\,\, y^{n-k}
         \ln \vartheta_1 \left(\pi \left[ y +\frac a c \right] ;q\right)\;.
\end{eqnarray}

The integrals that appear in (\ref{33}) and (\ref{34}) can be computed from the results of the following lemma.
\begin{lemma}
  For $0<\Im z <\pi \Im \tau$ with $q=e^{i\pi \tau}$, one has a Fourier-type representation
  \beq\label{35}
\ln \vartheta_1 (z;q) = \frac 1 6 \ln \, q +\ln \eta (\tau) + \ln (2\sin z)
 -2\sum_{n=1}^\infty \frac{q^{2n}} {1-q^{2n}} \frac{\cos (2nz)} n \;.
\label{fourier}
\eeq
\end{lemma}

\begin{proof}
  The logarithmic derivative of the first Jacobi theta function can be represented in terms of an infinite series for $0<\Im z <\pi \Im \tau$
  as \cite{gradshtein07,erdelyi53}
  \begin{equation}\label{36}
    \frac{\diff}{\diff z}\ln\vartheta_{1}(z;q)=\cot z+4\sum_{n=1}^{\infty}\frac{q^{2n}}{1-q^{2n}}\sin(2 n z)=\cot z+4\sum_{n=1}^{\infty}\frac{q^{2n}\sin(2z)}{1-2q^{2n}\cos(2z)+q^{4n}}\;.
  \end{equation}
  Anti-differentiation of (\ref{36}) yields
  \begin{equation}\label{37}
   \ln\vartheta_{1}(z;q)=f(q)+\ln\sin z-2 \sum_{n=1}^{\infty}\frac{q^{2n}}{1-q^{2n}}\frac{\cos(2 n z)}{n}=f(q)+\ln\sin z+\sum_{n=1}^{\infty}\ln\left(1-2q^{2n}\cos(2z)+q^{4n}\right)\;,
  \end{equation}
  where $f(q)$ is an arbitrary function. In order to determine the unknown $f(q)$ we use the infinite product representation \cite{gradshtein07,erdelyi53}
  \begin{equation}\label{38}
    \vartheta_{1}(z;q)=2G(q)q^{\frac{1}{4}}\sin z\prod_{n=1}^{\infty}\left(1-2q^{2n}\cos(2z)+q^{4n}\right)\;,
  \end{equation}
  where
  \begin{equation}\label{39}
    G(q)=\prod_{n=1}^{\infty}\left(1-q^{2n}\right)\;.
  \end{equation}
  From equation (\ref{38}) and by recalling that the Dedekind eta function is defined for $\Im\tau>0$ and $q=e^{i \pi \tau}$ as
  \begin{equation}\label{41}
    \eta(\tau)=e^{\frac{{\rm i}\pi\tau}{12}}\prod_{n=1}^{\infty}\left(1-q^{2n}\right)\; ,
  \end{equation}
one can easily find
  \begin{equation}\label{40}
    \ln\vartheta_{1}(z;q)=\ln 2+\ln \eta(\tau)+\frac{1}{6}\ln q+\ln\sin z+\sum_{n=1}^{\infty}\ln\left(1-2q^{2n}\cos(2z)+q^{4n}\right)\;.
  \end{equation}
By comparing (\ref{37}) with (\ref{40}) we finally find that
  \begin{equation}\label{42}
    f(q)=\ln 2+\ln \eta(\tau)+\frac{1}{6}\ln q\;.
  \end{equation}
  Substitution of the explicit expression (\ref{42}) for the function $f(q)$ in (\ref{37}) yields the claim (\ref{35}).
\end{proof}

The use of (\ref{fourier}) in the expression (\ref{33}) allows us to obtain the following result for the derivative of the doubly-periodic Barnes zeta function at $s=0$,
\begin{eqnarray}\label{43}
  \zeta ' (0,a,b,c)=-\ln \vartheta_1 \left(\pi \frac a   c ; q\right)+\frac{{\rm i}\pi}{2}+\frac{1}{6}\ln q+\ln\eta\left({\rm i}\frac{b}{c}\right)+\int_0^1 \ln \left(2\sin \left[\pi \left( y+\frac a c \right) \right] \right)\;.
\end{eqnarray}

When $n\geq 1$, we use once again the result (\ref{fourier}) in (\ref{34}). By performing this substitution and by recalling (\ref{31}) one obtains
\begin{eqnarray}\label{44}
  \zeta ' (-n,a,b,c)&=&\frac{(-1)^{n}c^{n}}{n+1}\sum_{k=0}^{n}\binom{n+1}{k}(n-k+1)B_{k}\left[\mathcal{I}_{n-k}\left(\frac{a}{c}\right)-2\sum_{l=1}^{\infty}\frac{q^{2l}} {1-q^{2l}}\frac{1}{l}\mathcal{J}_{n-k,l}\left(\frac{a}{c}\right)\right]\nonumber\\
  &+& {\rm i}\pi\frac{(-1)^{n}c^{n}}{n+1} B_{n+1},
\end{eqnarray}
where we have introduced for typographical convenience the functions
\begin{equation}\label{45}
  \mathcal{I}_{n-k}\left(\frac{a}{c}\right)=\int_{0}^{1}\diff y\,\, y^{n-k} \ln \left(2\sin \left[\pi \left( y+\frac a c \right) \right] \right)\; ,
\end{equation}
\begin{equation}\label{46}
   \mathcal{J}_{n-k,l}\left(\frac{a}{c}\right)=\int_{0}^{1}\diff y\,\, y^{n-k}\cos\left(2\pi l\left[y+\frac{a}{c}\right]\right)\;.
\end{equation}
Note that the integral appearing in (\ref{43}) reduces, according to the definition (\ref{45}), to $\mathcal{I}_{0}\left(a/c\right)$.

The integrals (\ref{45}) and (\ref{46}) can actually be explicitly computed in terms of polylogarithmic functions and trigonometric functions, respectively.
The following lemma provides an expression for the function $\mathcal{I}_{n-k}\left(a/c\right)$.
\begin{lemma}
  For $n\in\mathbb{N}_{0}$ and for $A\in\mathbb{C}\backslash\mathbb{Z}$,
  \begin{eqnarray}\label{48}
    \mathcal{I}_{n}\left(A\right)&=&-\frac{{\rm i}\pi n}{2(n+1)(n+2)}-\frac{{\rm i}\pi A}{n+1}\nonumber\\
    &+&\frac{(-A)^{n+1}}{n+1}\sum_{k=1}^{n+1}\binom{n+1}{k}(-1)^{k}\sum_{l=1}^{k}\frac{k!}{(k-l)!}\left(\frac{{\rm i}}{2\pi A}\right)^{l}\left(\sum_{j=1}^{k-l}\binom{k-l}{j}A^{-j}\right)\textrm{Li}_{l+1}\left(e^{2\pi{\rm i}A}\right)\;.
  \end{eqnarray}
\end{lemma}

\begin{proof}
Performing an integration by parts leads to the expression
  \begin{eqnarray}\label{49}
 \mathcal{I}_{n}\left(A\right)=\frac{\ln 2}{n+1}+\frac{{\rm i}\pi}{n+1}+\frac{1}{n+1}\ln\sin(\pi A)-\frac{\pi}{n+1}\int_{0}^{1}\diff y\,\, y^{n+1}\cot\left(\pi\left[y+A\right]\right)\;.
  \end{eqnarray}
  The integral in (\ref{49}) containing the cotangent can be explicitly computed in terms of polylogarithmic functions.
  First, by exploiting the change of variables $y+A\to y$, we obtain
  \begin{equation}\label{50}
    \int_{0}^{1}\diff y\,\, y^{n+1}\cot\left(\pi\left[y+A\right]\right)=\sum_{k=0}^{n+1}\binom{n+1}{k}(-A)^{n+1-k}\int_{A}^{A+1}\diff y\,\, y^{k}\cot\left(\pi y\right)\;,
  \end{equation}
  and by rewriting the cotangent in terms of complex exponentials the integral on the right-hand side of (\ref{50}) takes the form
\begin{equation}\label{51}
  \int_{A}^{A+1}\diff y\,\, y^{k}\cot\left(\pi y\right)=-{\rm i}\int_{A}^{A+1}\diff y\,\, y^{k}+2{\rm i}\int_{A}^{A+1}\diff y\,\, y^{k}\frac{e^{2\pi{\rm i}y}}{e^{2\pi{\rm i}y}-1}\;.
\end{equation}
This can be rewritten using the polylogarithmic function defined for $|z|< 1$ and $s\in\mathbb{C}$ by the sum
\begin{equation}\label{52}
  \textrm{Li}_{s}(z)=\sum_{n=1}^{\infty}\frac{z^{n}}{n^{s}}\;,
\end{equation}
and by analytic continuation in the entire complex $z$-plane \cite{lewin}. Of particular importance to our analysis is the following property satisfied by the polylogarithmic function \cite{lewin}
\begin{equation}\label{53}
  \frac{1}{2\pi{\rm i}}\frac{\diff}{\diff y}\textrm{Li}_{n}\left(e^{2\pi{\rm i}y}\right)=\textrm{Li}_{n+1}\left(e^{2\pi{\rm i}y}\right)\;,
\end{equation}
and since
\begin{equation}
\textrm{Li}_{0}\left(e^{2\pi{\rm i}y}\right)=-\frac{e^{2\pi{\rm i}y}}{e^{2\pi{\rm i}y}-1}\;,
\end{equation}
we have that
\begin{equation}\label{54}
  -\frac{1}{2\pi{\rm i}}\frac{\diff}{\diff y}\textrm{Li}_{1}\left(e^{2\pi{\rm i}y}\right)=\frac{e^{2\pi{\rm i}y}}{e^{2\pi{\rm i}y}-1}\;.
\end{equation}
The result (\ref{54}) employed in (\ref{51}) yields
\begin{equation}\label{55}
  2{\rm i}\int_{A}^{A+1}\diff y\,\, y^{k}\frac{e^{2\pi{\rm i}y}}{e^{2\pi{\rm i}y}-1}=-\frac{1}{\pi}\int_{A}^{A+1}\diff y\,\, y^{k}\frac{\diff}{\diff y}\textrm{Li}_{1}\left(e^{2\pi{\rm i}y}\right)\;.
\end{equation}
Integrating by parts $k$ times and using, at each step, the relation (\ref{53}) leads to
\begin{eqnarray}\label{56}
  -\frac{1}{\pi}\int_{A}^{A+1}\diff y\,\, y^{k}\frac{\diff}{\diff y}\textrm{Li}_{1}\left(e^{2\pi{\rm i}y}\right)&=&-\frac{1}{\pi}\sum_{n=0}^{k}\frac{k!}{(k-n)!}\frac{(-1)^{n}y^{k-n}}{(2\pi{\rm i})^{n}}\textrm{Li}_{n+1}\left(e^{2\pi{\rm i}y}\right)\Bigg|_{A}^{A+1}\nonumber\\
  &=&-\frac{1}{\pi}\sum_{n=0}^{k}\frac{k!}{(k-n)!}\frac{(-1)^{n}}{(2\pi{\rm i})^{n}}\sum_{l=1}^{k-n}\binom{k-n}{l}A^{k-n-l}\textrm{Li}_{n+1}\left(e^{2\pi{\rm i}A}\right)\;.
\end{eqnarray}
In light of the previous result and after computing the elementary integral on the right-hand side of (\ref{51})
we can conclude that
\begin{equation}\label{57}
  \int_{A}^{A+1}\diff y\,\, y^{k}\cot\left(\pi y\right)=-\frac{{\rm i}}{k+1}\sum_{l=1}^{k+1}\binom{k+1}{l}A^{k+1-l}-\frac{1}{\pi}\sum_{n=0}^{k}\frac{k!}{(k-n)!}\frac{(-1)^{n}}{(2\pi{\rm i})^{n}}\sum_{l=1}^{k-n}\binom{k-n}{l}A^{k-n-l}\textrm{Li}_{n+1}\left(e^{2\pi{\rm i}A}\right)\;.
\end{equation}
By substituting expression (\ref{57}) in (\ref{50}) we obtain, for (\ref{49}),
\begin{eqnarray}\label{58}
  \mathcal{I}_{n}\left(A\right)&=&\frac{\ln 2}{n+1}+\frac{{\rm i}\pi}{n+1}+\frac{1}{n+1}\ln\sin(\pi A)+\frac{{\rm i}\pi (-A)^{n+2}}{n+1}
  \sum_{k=0}^{n+1}\binom{n+1}{k}\frac{(-1)^{k+1}}{k+1}\sum_{l=1}^{k+1}\binom{k+1}{l}A^{-l}\nonumber\\
  &+&\frac{(-A)^{n+1}}{n+1}\sum_{k=0}^{n+1}\binom{n+1}{k}(-1)^{k}\sum_{l=0}^{k}\frac{k!}{(k-l)!}\left(\frac{{\rm i}}{2\pi A}\right)^{l}\left(\sum_{j=1}^{k-l}\binom{k-l}{j}A^{-j}\right)\textrm{Li}_{l+1}\left(e^{2\pi{\rm i}A}\right)\;.
\end{eqnarray}
The above result can be simplified further, in fact
\begin{equation}\label{58a}
  \frac{{\rm i}\pi (-A)^{n+2}}{n+1}
  \sum_{k=0}^{n+1}\binom{n+1}{k}\frac{(-1)^{k+1}}{k+1}\sum_{l=1}^{k+1}\binom{k+1}{l}A^{-l}=\frac{{\rm i}\pi (-A)^{n+2}}{n+1}\sum_{k=0}^{n+1}\binom{n+1}{k}\frac{(-1)^{k+1}}{k+1}\left[\left(1+\frac{1}{A}\right)^{k+1}-1\right]\;,
\end{equation}
and by using the relation, valid for $\alpha\in\mathbb{C}$ \cite{gradshtein07},
\begin{equation}\label{59}
  \sum_{k=0}^{m}\binom{m}{k}\frac{\alpha^{k+1}}{k+1}=\frac{(\alpha+1)^{m+1}-1}{m+1}\;,
\end{equation}
one obtains
\begin{equation}\label{60}
   \frac{{\rm i}\pi (-A)^{n+2}}{n+1}
  \sum_{k=0}^{n+1}\binom{n+1}{k}\frac{(-1)^{k+1}}{k+1}\sum_{l=1}^{k+1}\binom{k+1}{l}A^{-l}=\frac{{\rm i}\pi }{(n+1)(n+2)}\;.
\end{equation}
The $l=0$ contribution from the last term in (\ref{58}) can be simplified by following an argument similar to the one leading from (\ref{58a}) to (\ref{60}).
By noticing that \cite{lewin}
\begin{equation}\label{61}
  \textrm{Li}_{1}\left(z\right)=-\ln(1-z)\;,
\end{equation}
one obtains
\begin{eqnarray}\label{62}
  \frac{(-A)^{n+1}}{n+1}\sum_{k=0}^{n+1}\binom{n+1}{k}(-1)^{k}\sum_{j=1}^{k}\binom{k}{j}A^{-j}\textrm{Li}_{1}\left(e^{2\pi{\rm i}A}\right)=-\frac{{\rm i}\pi A}{n+1}-\frac{\ln 2}{n+1}-\frac{1}{n+1}\ln\sin(\pi A)-\frac{3\pi{\rm i}}{2(n+1)}\;.
\end{eqnarray}
By substituting  (\ref{60}) and (\ref{62}) in the expression (\ref{58}) the claim (\ref{48}) immediately follows.
\end{proof}
Let us now focus on the analysis of the function $\mathcal{J}_{n-k,l}(a/c)$ in (\ref{46}), which can be rewritten, after a suitable change of variables, as
\begin{equation}\label{63}
  \mathcal{J}_{n-k,l}\left(\frac{a}{c}\right)=\frac{1}{(2\pi l)^{k+1}}\cos\left(2\pi l \frac{a}{c}\right)\int_{0}^{2\pi l}\diff x\,x^{k}\cos x
  +\frac{1}{(2\pi l)^{k+1}}\sin\left(2\pi l \frac{a}{c}\right)\int_{0}^{2\pi l}\diff x\,x^{k}\sin x\;.
\end{equation}
The use of known reduction formulae for the trigonometric integrals on the right-hand side of (\ref{63}) yields
\begin{equation}\label{64}
  \int \diff x\, x^{k}\left\{\cos x \atop \sin x\right\}=\left\{\sin x \atop -\cos x\right\}\sum_{l=0}^{\left[\frac{k}{2}\right]}\frac{k!}{(k-2l)!}(-1)^{l}x^{k-2l}
  +\left\{\cos x \atop \sin x\right\}\sum_{l=0}^{\left[\frac{k-1}{2}\right]}\frac{k!}{(k-2l-1)!}(-1)^{l}x^{k-2l-1}\;,
\end{equation}
where $[x]$ denotes the integer part of $x$. The application of the expression (\ref{64}) to (\ref{63}) provides the result
\begin{eqnarray}\label{65}
  \mathcal{J}_{n-k,l}\left(\frac{a}{c}\right)&=&\frac{1}{(2\pi l)^{k+1}}\cos\left(2\pi l \frac{a}{c}\right)\left[\sum_{j=0}^{\left[\frac{k-1}{2}\right]}\frac{k!}{(k-2j-1)!}(-1)^{j}(2\pi l)^{k-2j-1}-\left(\frac{1+(-1)^{k+1}}{2}\right)(-1)^{\left[\frac{k-1}{2}\right]}k!\right]\nonumber\\
  &-&\frac{1}{(2\pi l)^{k+1}}\sin\left(2\pi l \frac{a}{c}\right)\left[\sum_{j=0}^{\left[\frac{k}{2}\right]}\frac{k!}{(k-2j)!}(-1)^{j}(2\pi l)^{k-2j}+\left(\frac{1+(-1)^{k}}{2}\right)(-1)^{\left[\frac{k}{2}\right]}k!\right]\;.
\end{eqnarray}

We can conclude that the equation (\ref{44}), together with the explicit expressions (\ref{48}) and (\ref{65}), gives a formula for computing the derivative
of the doubly-periodic Barnes zeta function at all negative integer points. Moreover, an expression for $\zeta ' (0,a,b,c)$ is obtained from (\ref{43}) by using (\ref{48}) with $n=0$.
The results obtained in this section allow for a very efficient way of computing $\zeta'(-n,a,b,c)$ as the explicit formulae can be easily implemented in an algebraic computer program.
For completeness, we display the expression for $\zeta'(-n,a,b,c)$ when $n=0$, $n=-1$, and $n=-2$.

For $n=0$, one has
\beq
\zeta ' (0,a,b,c)&=& -\ln \vartheta_1 \left(\pi \frac a   c ; q\right)
     +\frac 1 6 \ln q + \ln \eta \left({\rm i}\frac{b}{c}\right) +
      \pi {\rm i} \left( \frac 1 2 -\frac a c \right),\label{final}
\eeq
for $n=-1$, we obtain
\begin{eqnarray}
  \zeta ' (-1,a,b,c)=\frac{c}{2\pi {\rm i}}\textrm{Li}_{2}\left(e^{2\pi {\rm i}\frac{a}{c}}\right)-\frac{c}{\pi}\sum_{l=1}^{\infty}\frac{q^{2l}}{1-q^{2l}}\frac{1}{l^{2}}\sin\left(2\pi l \frac{a}{c}\right)\;,
\end{eqnarray}
and for $n=-2$, one gets
\begin{eqnarray}
\zeta ' (-2,a,b,c)=-\frac{c^{2}}{2\pi^{2} }\textrm{Li}_{3}\left(e^{2\pi {\rm i}\frac{a}{c}}\right)-\frac{c^{2}}{\pi^{2}}\sum_{l=1}^{\infty}\frac{q^{2l}}{1-q^{2l}}\frac{1}{l^{3}}\cos\left(2\pi l \frac{a}{c}\right)\;.
\end{eqnarray}
The above expressions are valid in the range of parameters stated below equation
(\ref{1}) and have to be periodically continued, as is clear from the fact
that the integral representation (\ref{pre}) is periodic.

\section{Concluding Remarks}

In this work we have performed a detailed analysis of the meromorphic extension of the doubly-periodic Barnes zeta function in the entire complex plane.
The contour integral representation (\ref{1c}) allowed us, after a suitable deformation of the integration contour, to obtain the meromorphic extension of
the doubly-periodic Barnes zeta function $\zeta(s,a,b,c)$. This analytically continued expression revealed to be particularly appropriate for the explicit computation of the values
$\zeta (n,a,b,c)$ and of the derivative
$\zeta'(-n,a,b,c)$ for $n\in\mathbb{N}_{0}$.

The process of analytic continuation delineated in this work is quite general and its applicability is not limited exclusively to the study of the doubly-periodic Barnes zeta function.
In fact, the method developed here can be applied to more general elliptic functions. Although the representation of elliptic functions in terms of integrals over $\mathbb{R}^{+}$ has been constructed, for instance in \cite{dienstfrey06}, our method would provide a contour integral representation for the class of elliptic functions. Since we have seen that such representation has some advantages, such as the double-periodicity of the results and the almost straightforward computation of the values and derivative at specific points, it could, perhaps, provide either new results or simplify already known ones regarding elliptic functions.
This seems to be an interesting idea which deserves further investigation.

Aside from their intrinsic mathematical interest, our results can find applications in problems related to physical systems. The Thirring model is used to describe simple interacting field theories \cite{thirring}.
The one-loop partition function for generalized Thirring models is proportional to the derivative at $s=0$ of the doubly-periodic Barnes zeta function \cite{sachs96}.
The result (\ref{final}) obtained here can then be directly applied to the analysis of these interacting field models.

\begin{acknowledgments}
KK acknowledges very fruitful discussions with Stuart Dowker on the subject.
\end{acknowledgments}

\end{document}